\documentclass[%
 aip,
 amsmath,amssymb,
 reprint,%
jmp]{revtex4-2}

\usepackage{graphicx}
\usepackage{dcolumn}
\usepackage{bm}
\usepackage{bbm}
\usepackage[mathlines]{lineno}

\usepackage{csquotes}
\usepackage{amsthm}
\usepackage{amsmath}
\usepackage{amssymb}
\usepackage{mathrsfs}%
\usepackage{mathtools}
\usepackage{mathptmx}
\usepackage{etoolbox}
\usepackage{xcolor}
\usepackage{cancel}
\usepackage{comment}

\makeatletter
\def\@email#1#2{%
 \endgroup
 \patchcmd{\titleblock@produce}
  {\frontmatter@RRAPformat}
  {\frontmatter@RRAPformat{\produce@RRAP{*#1\href{mailto:#2}{#2}}}\frontmatter@RRAPformat}
  {}{}
}%
\makeatother

\newtheorem{theorem}{Theorem}[section]
\newtheorem*{definition}{Definition}
\newtheorem{corollary}{Corollary}[theorem]
\newtheorem{proposition}{Proposition}[section]
\newtheorem{lemma}{Lemma}[section]

\begin{document}

\preprint{AIP/123-QED}

\title[Invariant Measures in Hamiltonian Systems: The Analytical Foundations of Statistical Physics]{Invariant Measures in Hamiltonian Systems: The Analytical Foundations of Statistical Physics}
\author{{Luis A.} {Cede\~no-P\'erez} and Alexis E. López-Velázquez}

\email{luisacp@ciencias.unam.mx, alexisvelazquez@ciencias.unam.mx}

\affiliation{ Instituto de Ciencias Nucleares, Universidad Nacional Aut\'onoma de M\'exico, AP  70543, Mexico City, Mexico }


\begin{abstract}
We construct a measure in the hamiltonian function level sets that is invariant under the hamiltonian flow for short times and flow preserving for arbitrarily long times. This allows a probabilistic approach to the study of hamiltonian systems, in the space of states with fixed energy. We prove that this measure generates the microcanonical partition function employed in physics and show that it can be transformed into the canonical partition function in an asymptotic limit, hence reproducing classical Statistical Physics. We also argue that this gives an alternative solution to Simon's second problem.
\end{abstract}



\maketitle

\tableofcontents

\section{Introduction}

\subsection{Probability in Hamiltonian Systems}

Hamiltonian systems constitute the basis for Classical Mechanics and are of great interest to both physicists and mathematicians, as shown by \cite{Arnold}, \cite{Marsden} and \cite{Saletan}. The interest in studying them through the lens of probability is twofold. First, despite the large quantity of desirable properties these systems possess, exact solutions can not always be found, hence the necessity of qualitative descriptions of their solutions. Probability is particularly well suited for this purpose, as instead of asking if all solutions of a system have a property we can determine the probability of solutions having said property. Second, probability in Hamiltonian systems lies at the heart of classical Statistical Physics, paving the way for the development of Thermodynamics. Hence, a solid mathematical foundation for probability in Hamiltonian systems results interesting for the mathematical study of Thermodynamics. Particularly, Geometric Thermodynamics has been studied in recent years in relation with General Relativity, in the form of thermodynamics of black holes or thermodynamics of the universe (see \cite{ANAYA} and \cite{LADINO}), as well as its relationship with Contact Geometry (see \cite{TermoContacto}). Probabilistic approaches to differential equations have become of interest in recent years due to the relation to Stochastic Differential Equations (see \cite{ProbODE}) and path integrals (see \cite{FuncMeas}), as well as to study the uncertainty introduced by numerical solutions (see \cite{ProbUncert}), among others.

It is not surprising then that the natural framework for developing probability measures for Hamiltonian systems is geometric (or geometric-measure-theoretical) in nature. For example, the Poincaré-Cartán invariant forms, which are simply the powers of the symplectic form (see \cite{Marsden}), generate invariant measures in even dimensional submanifolds. However, these invariants are not enough for the probabilistic description of hamiltonian systems, since they are either too coarse or too fine for this purpose. To see what we mean by this, let's state the principal problem we will consider in this work.

\subsection{Invariant Measures}

Consider a hamiltonian $H$ in $\mathbb{R}^{2n}$ and $\Phi_{t}$ the hamiltonian flow. If $H$ does not depend explicitly on time, the dynamics of the system take place in constant energy sets. If $H$ is smooth, then the constant energy sets (of regular values) are regular manifolds of dimension $2n - 1$. Hence, any acceptable probability measure must measure $2n - 1$ dimensional sets. Since the Poincaré-Cartán invariants are of even dimension, none of them are adequate for this purpose.

What other properties, other than dimensionality, should a probability measure have to be adequate to study probability in Hamiltonian systems? Hamiltonian systems describe the time evolution of a classical system, such evolution is represented precisely by the hamiltonian flow $\Phi_{t}$. Given any event $A$ (a Borel set) a probability measure $\mu$ assigns a probability $\mu(A)$ to the event. Since the hamiltonian flow determines the time evolution of the system, the event $A$ can be evolved in time to obtain an event $\Phi_{t}(A)$ with probability $\mu(\Phi_{t}(A))$. Since the hamiltonian does not depend explicitly on time, we expect the averaged or \textquote{macroscopic} quantities should not depend on time either, reflecting some sort of stationarity of probability. Hence, the most natural condition the measure $\mu$ should satisfy is
\begin{equation*}
    \mu(\Phi_{t}(A)) = \mu(A).
\end{equation*}
Measures that satisfy this condition are known as \textbf{invariant measures}. In the context of Thermodynamics, such measures represent that the system is in thermodynamic equilibrium. Mathematically, invariance simply requires that the probability of an event is independent of time and thus averages of quantities are real numbers instead of functions of time, respecting the temporal symmetry of the system.

In this work, we will construct an invariant measure $\omega_{E}$ in $H^{-1}(E)$, which we will call the \textbf{microcanonical measure}, which is sensible to $2n - 1$ dimensional events, and thus, solves the previous problem. The name \textquote{microcanonical} comes from Statistical Physics, as will be explained in the following subsection. The basis of our construction is Liouville's theorem, which states that the $2n$ dimensional Lebesgue measure $\lambda$ is invariant under the hamiltonian flow. We will use the Coarea Formula from Differential Geometry (or Geometric Measure Theory) to \textquote{differentiate} the Lebesgue measure and obtain the measure $\omega_{E}$. The formulas that make this precise are \eqref{Eqdefomega} and \eqref{MicroDerivada}. Invariant measures have also been studied for dissipative dynamical systems (see \cite{InvDissipative}), in which the previous considerations do not apply. Our study of invariance will be divided into results of short and long times. Short times are those in which the hamiltonian flow is a diffeomorphism and long times are those in which said condition may not hold anymore.

\subsection{The Foundations of Statistical Physics}

Classical Statistical Physics is based on the quantity $\Omega(E)$ known as the \textbf{microcanonical partition function}, which measures the \textquote{number of states} with energy $E$. Botlzmann's Principle then states that the entropy $S$ of the system can be computed as
\begin{equation*}
    S = k_{B}\;\ln\;\Omega(E),
\end{equation*}
where $k_{B}$ is a constant known as Boltzmann's constant. The intuition is that the more microscopic states of energy $E$ there are, the less likely it is that the system takes another macroscopic state. Hence, the entropy $S$ is unlikely to decrease, preserving the second law of Thermodynamics, at least on average.

The problem is then how to compute $\Omega(E)$. Since $\Omega(E)$ measures the number of states with energy $E$ the usual expressions found on physics literature are either
\begin{equation*}
    \Omega(E) = \int \delta(E - H(q,p))\;dq\;dp
\end{equation*}
or
\begin{equation*}
    \Omega(E) = \int \, \mathbbm{1}_{\{H(q,p)=E\}}\;dq\;dp.
\end{equation*}
While the intent of both expressions is clear, none of them are properly defined in a mathematically rigorous sense. See \cite{LegendreTermo} for a summary of classical Statistical Physics or \cite{Pathria} for a detailed account.

While the physical foundations of Statistical Physics have been widely discussed and studied (see \cite{FoundationSP1} and \cite{FoundationSP2}) their mathematical foundations are not so well established. For example, in \cite{Simon} it is argued that most formulations of classical Statistical Mechanics require the Ergodic Hypothesis, which is known not to hold in many classical systems. This is part Simon's second problem, which we will discuss in the following subsection.

Our construction of the microcanonical measure $\omega_{E}$ will provide a rigorous expression for the microcanonical partition function and set a rigorous foundation for classical Statistical Physics. We will show that if we define $\Omega(E)$ as the normalization constant required to make $\omega_{E}$ a probability measure then $\Omega(E)$ will reproduce classical Statistical Physics. To show this we will transform Boltzmann's Principle to the Helmholtz free energy representation. The required result is that if $F$ is the Helmholtz free energy of the system then
\begin{equation*}
    F = -k_{B}T\;\ln\;\mathcal{Z},
\end{equation*}
where $T$ is the temperature and
\begin{equation*}
    \mathcal{Z} = \int e^{-\beta H(q,p)}\;dq\;dp
\end{equation*}
is known as the \textbf{canonical partition function}. The equality is not exact, but rather an asymptotic equality. This is enough to show that the theory we will develop coincides with classical Statistical Physics, since most physical results are obtained from the transformed Boltzmann Principle. Thus our construction is not only mathematically relevant, but also gives a mathematically rigorous foundation for classical Statistical Physics.

\subsection{Simon's Second Problem}

In \cite{Simon} Barry Simon listed 15 problems he considered fundamental for Mathematical Physics to solve. Said problems range from old problems of classical theories, like Classical Mechanics, Thermodynamics, and Fluid Dynamics, to more famous problems of the younger disciplines of Quantum Mechanics and General Relativity. Most of Simon's problems remain unsolved to this day, since only problems 4, 5 and 7 have been solved to this day (see \cite{Puig}, \cite{Avila} and \cite{Kiselev}).

In the statement of his second problem, Simon argues that classical Statistical Physics is based on the Ergodic Hypothesis, which is not satisfied in many systems, as shown by KAM Theory. The problem is then to establish ergodicity for certain systems in order to formulate Statistical Physics correctly. Our work shows that classical Statistical Physics is well formulated whether the Ergodic Hypothesis holds or not, hence, offers a solution to Simon's core problem. We do not solve the specific problem, ergodicity of certain systems, however, this problem arises from the intent of formulating Statistical Mechanics through the Ergodic Hypothesis, hence the main purpose is actually to give a sound foundation to classical Statistical Physics.

\section{Invariant Measures in Hamiltonian Systems}

In this section we construct the microcanonical measure $\omega_{E}$, show that it is invariant under the hamiltonian flow and develop simple methods for the computation of its normalization constant $\Omega(E)$. Invariance results will be divided in those for short times, in which the hamiltonian flow is a diffeomorphism, and long times, in which the hamiltonian flow may fail to be a diffeomorphism.

\subsection{Microcanonical Measure}\label{SubsecMicro}

Consider the space $\mathbb{R}^{2n}$ consisting of points of the form $(q,p)$, where $q,p\in\mathbb{R}^{n}$. Let $H\colon \mathbb{R}^{2n} \to \mathbb{R}$ be a $C^{2}$ hamiltonian function. We consider the stationary hamiltonian system defined by $H$, that is,
\begin{align*}
    \dot{q} &= \partial_{p}H\\
    \dot{p} &= -\partial_{q}H.
\end{align*}
The local flow of the previous system is the \textbf{hamiltonian flow}, denoted by $\Phi_{t}$. By the fundamental theorem of flows (see \cite{LeeDif}), there exists a symmetric neighborhood of zero $(-\epsilon,\epsilon)$ with $\epsilon > 0$ where the flow is well-defined, such that $\{\Phi_t \mid t \in (-\epsilon, \epsilon)\}$ forms a family of diffeomorphisms.

If one assumes the level sets of $H$ to be compact (which happens if $H$ is bounded from below and coercive) then $\Phi_{t}$ is defined for all times. We wil not make this assumption throughout the work.

Since we are assuming that $H$ does not depend explicitly on time, the conservation of energy is valid, which means that the dynamics of the system take place in constant energy sets $H^{-1}(E)$. The conservation of energy can be written in terms of the hamiltonian flow and the constant energy sets as
\begin{equation}\label{eq:cons.energia}
    \Phi_t(H^{-1}(E))\subseteq H^{-1}(E).
\end{equation}
In order to do probability on the stationary hamiltonian system we require that the probability of an event does not depend on time, hence we make the following definition.

\begin{definition}
Let $\mu$ be a measure in the Borel subsets of $H^{-1}(E)$. We will say that $\mu$ is \textbf{invariant under the hamiltonian flow} if for any Borel set $A$ and every time $t$ the equation
\begin{equation*}
    \mu(\Phi_{t}(A)) = \mu(A)
\end{equation*}
is satisfied whenever $\Phi_{t}(A)$ is also a Borel set. We say that $\mu(A)$ is invariant under the hamiltonian flow if
\begin{equation*}
    \mu(\Phi_{t}(A)) = \mu(A)
\end{equation*}
for every $t$ such that $\Phi_{t}(A)$ is a Borel set.
\end{definition}

Note that for short times, that is $t\in (-\epsilon,\epsilon)$, the function $\Phi_{t}$ is a diffeomorphism. Hence, $\Phi_{t}(A)$ is always a Borel set whenever $A$ is a Borel set, thus this condition is redundant in this case and only the equaility of measures needs to be verified. For long times this may fail, leading to the parallel concept of flow preservation, which will be discussed in the following section. For this reason this subsection deals only with short times.

The objetive is to construct an invariant measure under the hamiltonian flow in the constant energy sets that quantifies the microstates $(q,p)$ contained in each Borel subset. The basic result in this direction is the following.

\begin{theorem}[Liouville's Theorem]\label{TeoLiouville}
The Lebesgue measure $\lambda$ of $\mathbb{R}^{2n}$ is invariant under the hamiltonian flow, in the following sense: for every $A$ Borel set of $\mathbb{R}^{2n}$ the equation
\begin{equation*}
    \lambda(\Phi_{t}(A)) = \lambda(A)
\end{equation*}
is satisfied for every time $t$ such that $\Phi_{t}(A)$ is a Borel set.
\end{theorem}

See \cite{Arnold} for the classical proof, \cite{LeeDif} or \cite{Marsden} for geometric proofs and \cite{LeonSimon} for the Geometric Measure Theory proof. The previous result seemingly trivializes the posed problem; since $H^{-1}(E)$ is a Borel set, the Borel sets of $H^{-1}(E)$ are once again Borel sets of $\mathbb{R}^{2n}$, hence the restricted measure $\lambda\restriction_{Bo(H^{-1}(E))}$ is an invariant measure on $H^{-1}(E)$. However, if $E$ is a regular value of $H$ then $H^{-1}(E)$ is a codimension $1$ submanifold of $\mathbb{R}^{2n}$, hence
\begin{equation*}
    \lambda(H^{-1}(E)) = 0,
\end{equation*}
thus, if $A\in Bo(H^{-1}(E))$ then
\begin{equation*}
    \lambda(A) \leq \lambda(H^{-1}(E)) = 0,
\end{equation*}
and the restricted measure $\lambda\restriction_{Bo(H^{-1}(E))}$ is the trivial measure.

The required invariant measure must then be finer than the Lebesgue measure and need only be defined in $H^{-1}(E)$. Liouville's Theorem, although flawed for our purpose, will be instrumental for proving invariance of the following measure.

\begin{definition}
Let $E\in\mathbb{R}$ be a regular value of $H$, hence $H^{-1}(E)$ is a codimension $1$ submanifold with volume measure $\mu_{E}$ (see \cite{LeeDif} and \cite{LeeGeo}). We define
\begin{equation*}
    \Omega(E) = \int_{H^{-1}(E)} \frac{d\mu_{E}}{|\nabla H|}.
\end{equation*}
If $\Omega(E) < \infty$, we define the \textbf{microcanonical measure} $\omega_{E}$ in the Borel subsets of $H^{-1}(E)$ as
\begin{equation}\label{Eqdefomega}
    \omega_{E}(A) = \frac{1}{\Omega(E)} \int_{A} \frac{d\mu_{E}}{|\nabla H|}.
\end{equation}
\end{definition}

It is immediate that $\omega_{E}$ is indeed a measure, and is actually a probability measure. Rademacher's theorem (see \cite{Evans}) implies that the same definition can be done if $H$ is only a Lipschitz function, replacing the volume measure $\mu_{E}$ by a Hausdorff measure. For the sake of simplicity, however, we will assume that $H$ is $C^{2}$ instead of only Lipschitz. Also note that $\omega_{E}$ is not generated by a Poincaré-Cartán Invariant (see \cite{Marsden}) since $\omega_{E}$ is generated by an odd degree form.

Note that the non-normalized microcanonical measure is trivial if and only if the normalization constant $\Omega(E)$ is null. If $H$ is smooth then this can only happen if $\mu_{E}(H^{-1}(E)) = 0$. Hence we will assume that $0 < \mu_{E}(H^{-1}(E)) < \infty$ and $0 < \Omega(E) < \infty$. The normalization constant $\Omega(E)$ will be very important later on.

Our main interest is to prove that $\omega_{E}$ is invariant under the hamiltonian flow, thus providing the desired measure. In this subsection we will work the case for short times, that is, when $t\in (-\epsilon,\epsilon)$, in which case $\Phi_{t}$ is a diffeomorphism. We leave the case of long times for the following subsection. The key result for establishing invariance is the following.

\begin{theorem}[Coarea Formula]\label{CoareaFormula}
If $f\colon \mathbb{R}^{2n} \to \mathbb{R}$ is a Lebesgue integrable function then
\begin{equation*}
    \int_{H^{-1}([a,b])}f\;d\lambda = \int_{a}^{b}\left(\int_{H^{-1}(E)}\frac{f}{|\nabla H|}\;d\mu_{E}\right)\;dE.
\end{equation*}
\end{theorem}

Once again, $H$ may only only be a Lipschitz function and the Coarea Formula remains valid with a Hausdorff measure in place of $\mu_{E}$. See \cite{Chavel} or \cite{Inductive} for proofs of the smooth case and \cite{Evans} or \cite{Inductive} for the Hausdorff measure result.

We now differentiate the Coarea Formula. To this end we now consider that $H$ is bounded below and has $0$ as its minimum. This is only to simplify notation and is not a relevant restriction (the following can be done by selecting any $e_{0}$ such that $e_{0} < E$ and replacing $0$ with $e_{0}$). We define the \textbf{subenergy set} $H_{E}$ as
\begin{equation*}
    H_{E} = H^{-1}([0,E]).
\end{equation*}
Differentiating the Coarea Formula an renaming variables we find that
\begin{equation*}
    \int_{H^{-1}(E_{0})}\frac{f}{|\nabla H|}\;d\mu_{E_{0}} = \frac{d}{dE}\left(\int_{H_{E}}f\;d\lambda\right)\Bigg|_{E = E_{0}}.
\end{equation*}
By choosing $f = 1$ and restricting $H$ to an open subset $A = \mathcal{A} \cap H^{-1}(E)$ of $H^{-1}(E)$, with $\mathcal{A}$ open in $\mathbb{R}^{2n}$, we find that the equation
\begin{equation*}
    \int_{A} \frac{d\mu_{E_{0}}}{|\nabla H|} = \frac{d}{dE}\lambda(\mathcal{A}\cap H_{E})\bigg|_{E = E_{0}}
\end{equation*}
is valid for open sets.

In what follows, for simplicity, we will work with the non-normalized microcanonical measure
\begin{equation*}
    \omega_{E}(A) = \int_{A} \frac{d\mu_{E}}{|\nabla H|},
\end{equation*}
since normalization does not alter invariance. Hence, the differentiation of the Coarea Formula actually yields an alternative expression for the microcanonical measure of open sets:
\begin{equation}\label{MicroDerivada}
    \omega_{E_{0}}(A) = \frac{d}{dE}\lambda(\mathcal{A}\cap H_{E})\bigg|_{E = E_{0}}.
\end{equation}
This will be the fundamental formula for proving invariance under the hamiltonian flow.

As a last preliminary before proving invariance, we have the following result.

\begin{lemma}\label{lemaAbreInt}
Let $A\subseteq H^{-1}(E_0)$ be a relative open set with $A=\mathcal{A} \cap H^{-1}(E_0)$, where $\mathcal{A}$ is open in $\mathbb{R}^{2n}$. For every $t\in(-\varepsilon,\varepsilon)$ the equation
\begin{equation*}
    \Phi_t(A)=\Phi_t(\mathcal{A})\cap H^{-1}(E_0)
\end{equation*}
is satisfied.
\end{lemma}
\begin{proof}
Let $t\in (-\varepsilon,\varepsilon)$. We have that
\begin{align*}
    \Phi_t(A)&=\Phi_t(\mathcal{A} \cap H^{-1}(E_0))\\
    &\subseteq\Phi_t(\mathcal{A})\cap \Phi_t(H^{-1}(E_0))\\
    &\subseteq \Phi_t(\mathcal{A})\cap H^{-1}(E_0),
\end{align*}
where the last inclusion follows from the conservation of energy \eqref{eq:cons.energia}.

For the remaining inclusion consider $x\in \Phi_t(\mathcal{A})\cap H^{-1}(E_{0})$. Thus there exists $y\in \mathcal{A}$ such that $x = \Phi_{t}(y)$. Since $A=\mathcal{A} \cap H^{-1}(E_{0})$ it is enough to show that $y\in H^{-1}(E_{0})$ to conclude that $y\in A$. Conservation of energy (\ref{eq:cons.energia}) implies that
\begin{align*}
    H(y) &= H(\Phi_{t}(y))\\
    &= H(x)\\
    &= E_{0},
\end{align*}
where the last equality follows from $x\in H^{-1}(E_{0})$. it follows that $y\in H^{-1}(E_{0})$ and thus $y\in A$. We conclude that $x = \Phi_{t}(y) \in \Phi_{t}(A)$ and obtain the remaining inclusion.
\end{proof}

We can now prove flow invariance of open sets for short times.

\begin{theorem}[Invariance of Open Sets]\label{teo.inv.abiertos}
Let $A\subseteq H^{-1}(E_0)$ be a relative open set, with $A=\mathcal{A} \cap H^{-1}(E_0)$ where $\mathcal{A}$ is open in $\mathbb{R}^{2n}$ is open. For every $t\in(-\varepsilon,\varepsilon)$ we have that
\begin{equation*}
    \omega_{E_0}(\Phi_t(A))=\omega_{E_0}(A).
\end{equation*}
Hence, $\omega_{E_{0}}(A)$ is invariant under the hamiltonian flow.
\end{theorem}
\begin{proof}
Let $t\in(-\varepsilon,\varepsilon)$. The fact that $\Phi_t$ is a diffeomorphism with inverse $\Phi_{-t} $ implies that for every $x\in H^{-1}(E_0)$, we have that
\begin{equation*}
    x=\Phi_t(\Phi_{-t}(x)).
\end{equation*}
Conservation of energy \eqref{eq:cons.energia} implies that 
\begin{align*}
    H(\Phi_{-t}(x))&=H(x)\\
    &=E_0.
\end{align*}
Thus $x\in \Phi_t(H^{-1}(E_0))$, hence $$ H^{-1}(E_0)\subseteq \Phi_t(H^{-1}(E_0)).$$
This last inclusion, along with the conservation of energy (\ref{eq:cons.energia}), implies that
\begin{equation*}
    \Phi_t(H^{-1}(E_0))=H^{-1}(E_0).
\end{equation*}
Using the the previous equality and the fact that
\begin{equation*}
    H_{E} = \bigcup_{e\in [0,E]} H^{-1}(e)
\end{equation*}
we find that
\begin{equation*}
    \Phi_t(H_E)= H_E.
\end{equation*}
Lemma \ref{lemaAbreInt} and the fact that $\Phi_t$ is a diffeomorphism imply that  $\Phi_t(A)$ is a reletive open in $H^{-1}(E_0)$. It follows from equation \eqref{MicroDerivada} and Liouville's Theorem \ref{TeoLiouville} applied to the set $\mathcal{A}\cap H_E$ that
\begin{align*}
    \omega_{E_0}(\Phi_t(A))&=\int_{\Phi_t(A)}  \frac{d\mu_{E}}{|\nabla H|}\\
    &= \frac{d}{dE} \lambda(\Phi_t(\mathcal{A})\cap H_E)\big|_{E=E_0}\\
    &= \frac{d}{dE} \lambda(\Phi_t(\mathcal{A})\cap\Phi_t(H_E))\big|_{E=E_0}\\
    &= \frac{d}{dE} \lambda(\Phi_t(\mathcal{A}\cap H_E))\big|_{E=E_0}\\
    &=\frac{d}{dE} \lambda(\mathcal{A}\cap H_E)\big|_{E=E_0}\\
    &=\omega_{E_0}(A),
\end{align*}
which is the required equality.
\end{proof}

The complete invariance of $\omega_{E_0}$ now follows from a simple application of 
Dynkin's $\pi$-$\lambda$ theorem, from which it can be shown that if $\mu$ and 
$\nu$ are two finite Borel measures on a topological space that coincide on open 
sets, then $\mu = \nu$.

\begin{theorem}[Invariance of Microcanonical Measure for Short Times]\label{InvarianceShortTime}
The microcanonical measure is invariant under the hamiltonian flow for $t\in (-\epsilon,\epsilon)$.
\end{theorem}
\begin{proof}
Let $t\in (-\varepsilon,\varepsilon)$ and $ A\subseteq H^{-1}(E_0)$ be a Borel measurable set. Since $\Phi_{-t} \restriction_{H^{-1}(E_0)}$ is continuous, the pushforward measure $(\Phi_{-t})_{*}\omega_{E_0}$ given by
\begin{align*}
    (\Phi_{-t})_{*}\omega_{E_0}(A)&=\omega_{E_0}((\Phi_{-t})^{-1}(A))\\
    &=\omega_{E_0}(\Phi_{t}(A))
\end{align*}
is a Borel measure. Invariance of open sets \ref{teo.inv.abiertos}, establishes that the measures $(\Phi_{-t})_{*}\omega_{E_0}$ and $\omega_{E_0}$ agree on open sets of $H^{-1}(E_0)$. Since both measures are Borel measures defined in a finite space, the fact that they agree in open sets implies that they actually are the same measure, thus
\begin{equation*}
    \omega_{E_0}(\Phi _t(A))=\omega_{E_0}(A).
\end{equation*}
\end{proof}

\subsection{Invariance for long Times}

For long times the hamiltonian flow may no longer be a diffeomorphism, hence $\Phi_{t}(A)$ may fail to be a Borel set, even if $A$ is a Borel set. Hence, the notion of flow invariance is not entirely natural within the framework of Measure Theory. This leads to the concept of flow preservation.

\begin{definition}
Let $\mu$ be a measure in $H^{-1}(E_{0})$. We say that the $\mu$ \textbf{preserves the hamiltonian flow} if
\begin{equation*}
    \mu(A) = \mu((\Phi_{t})^{-1}(A))
\end{equation*}
for every Borel set $A$ and every time $t$ in the domain of $\Phi$. More simply,
\begin{equation*}
    (\Phi_{t})_{\ast}\mu = \mu
\end{equation*}
for every $t$ in the domain of $\Phi$.
\end{definition}

Note that if every $\Phi_{t}$ is a diffeomorphism, then flow preservation and flow invariance are equivalent. In this subsection we first show that the microcanonical measure is flow preserving. Then we show that if further regularity properties are imposed on the hamiltonian then the microcanonical measure is also flow invariant.

Flow preservation follows through similar techniques of those of the previous section.

\begin{theorem}[Flow Preservation for long Times]
The microcanonical measure preserves the hamiltonian flow for arbitrarily long times.
\end{theorem}
\begin{proof}
The strategy is to prove flow preservation for open sets and then concluding the general case by the same Dynkin's $\pi$-$\lambda$ theorem 
argument as before. Let $A$ be a relative open set of $H^{-1}(E_{0})$, thus there exists $\mathcal{A}$ open in $\mathbb{R}^{2n}$ such that $A = \mathcal{A} \cap H^{-1}(E_{0})$. For $u\in [0,E]$ with $E> E_0$, we define
\begin{equation*}
    A_{u} = \mathcal{A}\cap H^{-1}(u),
\end{equation*}
hence $A_{E_{0}} = A$. We have that
\begin{align*}
    (\Phi_{t})^{-1}(A_{u}) &= (\Phi_{t})^{-1}(\mathcal{A} \cap H^{-1}(u))\\
    &= (\Phi_{t})^{-1}(\mathcal{A}) \cap (\Phi_{t})^{-1}(H^{-1}(u))\\
    &\subseteq (\Phi_{t})^{-1}(\mathcal{A}) \cap H^{-1}(u),
\end{align*}
where the last inclusion follows from conservation of energy \eqref{eq:cons.energia}. We now show that both of the last sets are actually equal. Given $x \in (\Phi_{t})^{-1}(\mathcal{A}) \cap H^{-1}(u)$ we have that $\Phi_{t}(x) \in \mathcal{A}$ and $H(x) = u$. We need to show that $x\in (\Phi_{t})^{-1}(H^{-1}(u))$, which is equivalent to $H(\Phi_{t}(x)) = u$, which follows from the fact that $H(x) = u$ and conservation of energy. Thus we conclude that
\begin{align}\label{EqLargeu}
    (\Phi_{t})^{-1}(A_{u}) = (\Phi_{t})^{-1}(\mathcal{A}) \cap H^{-1}(u),
\end{align}
and, in particular,
\begin{equation}\label{eqLargeTime1}
    (\Phi_{t})^{-1}(A) = (\Phi_{t})^{-1}(\mathcal{A}) \cap H^{-1}(E_{0}).
\end{equation}
Taking the union over $u\in [0,E]$ in \eqref{EqLargeu}, it also follows that
\begin{equation*}
    (\Phi_{t})^{-1}(\mathcal{A})\cap H_{E} = (\Phi_{t})^{-1}(\mathcal{A} \cap H_{E})
\end{equation*}

Since $\Phi_{t}$ is continuous, the set $(\Phi_{t})^{-1}(\mathcal{A})$ is open, hence, the equality in \eqref{eqLargeTime1} is a valid decomposition for \eqref{MicroDerivada}. It follows that
\begin{align*}
    \omega_{E_{0}}((\Phi_{t})^{-1}(A)) &= \frac{d}{dE}\lambda((\Phi_{t})^{-1}(\mathcal{A}) \cap H_{E})\bigg|_{E = E_{0}}\\
    &= \frac{d}{dE}\lambda((\Phi_{t})^{-1}(\mathcal{A} \cap H_{E}))\bigg|_{E = E_{0}}\\
    &= \frac{d}{dE}\lambda(\mathcal{A} \cap H_{E})\bigg|_{E = E_{0}}\\
    &= \omega_{E_{0}}(A).
\end{align*}
It follows that $\omega_{E_{0}}$ preserves the hamiltonian flow on open sets. By the same Dynkin's $\pi$-$\lambda$ theorem argument as in Theorem \ref{InvarianceShortTime}, we conclude that $\omega_{E_{0}}$ is flow preserving.
\end{proof}

\subsection{Microcanonical Partition Function}

We now study the normalization constant $\Omega(E)$ introduced in subsection \ref{SubsecMicro}. Note that $\Omega(E)$ is defined whenever $E$ is a regular value of $H$. Sard's Theorem implies that the set of critical values of $H$ is a Lebesgue null set, hence $\Omega(E)$ is defined for almost every $E$. This allows us to make the following definition.

\begin{definition}
We define the \textbf{microcanonical partition function} $\Omega$ as
\begin{equation*}
    \Omega(E) = \int_{H^{-1}(E)} \frac{d\mu_{E}}{|\nabla H|}.
\end{equation*}
\end{definition}

The computation of the microcanonical partition function is fundamental for computing the microcanonical measure. This section is devoted to different techniques for computing $\Omega$. The results of this section will also be used to show that the microcanonical measure reproduces classical Statistical Physics.

The simplest technique for computing $\Omega$ is based on the following function.

\begin{definition}
We define the \textbf{state density function} $g\colon \mathbb{R}^{+} \to \mathbb{R}$ as
\begin{equation*}
    g(E) = \lambda(H_{E}).
\end{equation*}
\end{definition}

Equation \eqref{MicroDerivada} implies that
\begin{equation*}
    \frac{dg}{dE} = \Omega.
\end{equation*}
Since $g$ is much easier to compute that $\Omega$, this is a reliable method for computing the microcanonical measure.

We now turn to the question of computing $\Omega$ for compound systems in terms of the partition function of its components. Consider two hamiltonians
\begin{equation*}
    H_{1}\colon \mathbb{R}^{2n_{1}} \to \mathbb{R}
\end{equation*}
and
\begin{equation*}
    H_{2}\colon \mathbb{R}^{2n_{2}} \to \mathbb{R},
\end{equation*}
from which we construct a third hamiltonian
\begin{equation*}
    H\colon \mathbb{R}^{2n_{1}}\times \mathbb{R}^{2n_{2}} \to \mathbb{R}
\end{equation*}
given by
\begin{equation*}
    H(q_{1},p_{1},q_{2},p_{2}) = H_{1}(q_{1},p_{1}) + H_{2}(q_{2},p_{2}),
\end{equation*}
which we may write more simply as
\begin{equation*}
    H = H_{1} + H_{2}.
\end{equation*}
Let $\Omega_{1}$, $\Omega_{2}$ and $\Omega$ be the microcanonical partition functions corresponding to $H_{1}$, $H_{2}$ and $H$, respectively. The objective is to write $\Omega$ in terms of $\Omega_{1}$ and $\Omega_{2}$. We will show that
\begin{equation}\label{EqConcolutivity}
    \Omega = \Omega_{1}\ast\Omega_{2},
\end{equation}
hence $\Omega$ is \textbf{convolutive on subsystems}. By general properties of the Laplace transform, it is enough to show that the laplace transform of $\Omega$ is multiplicative on subsystems.

\begin{definition}
We define the \textbf{canonical partition function} $\mathcal{Z}$ as the Laplace transform of $\Omega$, that is,
\begin{equation*}
    \mathcal{Z}(\beta) = \int_{0}^{\infty}\Omega(E)e^{-\beta E}\;dE.
\end{equation*}
\end{definition}

Thus, our objective is to show that
\begin{equation*}
    \mathcal{Z} = \mathcal{Z}_{1}\cdot\mathcal{Z}_{2}.
\end{equation*}
We begin with a very general proposition.

\begin{proposition}
\label{prop:integrar f(H)}
If $f\colon \mathbb{R} \to \mathbb{R}$ is such that $f\circ H$ is integrable with respect to $\lambda$, then the integral over the entire domain of H satisfies that
\begin{equation*}
    \int f\circ H \, d\lambda =\int_{0}^{\infty} f(E) \Omega(E) \, dE. 
\end{equation*} 
\end{proposition}
\begin{proof}
For $E\geq 0$ and $x\in H^{-1}(E)$ we have $f\circ H(x) = f(E)$, thus the Coarea Formula \ref{CoareaFormula} implies that
\begin{align*}
    \int f\circ H \,d\lambda &= \int_{0}^{\infty} \int_{H^{-1}(E)} \frac{f\circ H}{|\nabla H |}\,d\mu_{E} \,dE \\
    &= \int_{0}^{\infty} \int_{H^{-1}(E)} \frac{f(E)}{|\nabla H |}\,d\mu_{E} \,dE\\
    &= \int_{0}^{\infty} f(E) \left( \int_{H^{-1}(E)} \frac{d\mu_{E}}{|\nabla H |}\right) \,dE \\
    &=\int_{0}^{\infty} f(E) \Omega(E) \, dE.
\end{align*}
\end{proof}

Hence, $\Omega$ behaves as a probability distribution for the computation of averages with respect to the Lebesgue measure. Although this proposition is interesting in its own right, it also gives an alternative formula for computing $\mathcal{Z}$.

\begin{corollary}
The canonical partition function of a hamiltonian $H$ can be computed as
\begin{align}
    \mathcal{Z}(\beta) =\int  e^{-\beta H(q,p)}\,d\lambda
\end{align}
\end{corollary}
\begin{proof}
Let $f(E)=e^{-\beta E}$, then by definition
\begin{equation*}
    \mathcal{Z}=\int_{0}^{\infty} f(E)\Omega(E)\;dE.
\end{equation*}
The previous proposition implies that
\begin{align*}
    \int_{0}^{\infty} f(E)\Omega(E) \,dE &= \int f\circ H \,d\lambda\\
    &=\int e^{-\beta H}\,d\lambda,
\end{align*}
from which the result follows.
\end{proof}

The multiplicativity of $\mathcal{Z}$ now follows almost immediately.

\begin{proposition}
If $H_{1}$, $H_{2}$ are two hamiltonians and $H = H_{1} + H_{2}$, then the canonical partition functions $\mathcal{Z}_{1}$, $\mathcal{Z}_{2}$ and $\mathcal{Z}$ corresponding to $H_{1}$, $H_{2}$ and $H$ satisfy
\begin{equation*}
    \mathcal{Z} = \mathcal{Z}_{1}\cdot\mathcal{Z}_{2}.
\end{equation*}
\end{proposition}
\begin{proof}
It follows from Fubini's Theorem and the previous expression for the canonical partition function that
\begin{align*}
    \mathcal{Z} &=\iint e^{-\beta H_{12}(q_1,p_1,q_2,p_2)} \,d\lambda(q_1,p_1) \,d\lambda(q_2,p_2)\\
    &= \iint e^{-\beta (H_1(q_1,p_1)+H_2(q_2,p_2))} \,d\lambda(q_1,p_1) \,d\lambda(q_2,p_2)\\
    &= \int e^{-\beta H_2(q_2,p_2)} \int e^{-\beta H_1(q_1,p_1)} \, d\lambda(q_1,p_1) \,d\lambda(q_2,p_2)\\
    &= \int e^{-\beta H_1(q_1,p_1)} \, d\lambda(q_1,p_1) \int  e^{-\beta H_2(q_2,p_2)} \,d\lambda(q_2,p_2) \\
    &= \mathcal{Z}_1\cdot \mathcal{Z}_2,
\end{align*}
\end{proof}

Equation \eqref{EqConcolutivity} follows immediately from this.

\begin{corollary}
If $H_{1}$, $H_{2}$ are two hamiltonians and $H = H_{1} + H_{2}$, then the microcanonical partition functions $\Omega_{1}$, $\Omega_{2}$ and $\Omega$ corresponding to $H_{1}$, $H_{2}$ and $H$ satisfy
\begin{equation*}
    \Omega = \Omega_{1}\ast\Omega_{2}.
\end{equation*}
\end{corollary}

\section{Transformation of Boltzmann's Principle}

In this section we show that our definition for $\Omega(E)$ reproduces classical Statistical Physics. To do this we transform Boltzmann's Principle and show that the classical expression \eqref{BoltzmannTrans} is valid as an asymptotic approximation, in agreement with classical Statistical Physics.

\subsection{Statement and Definitions}

Boltzmann's Principle states that the \textbf{entropy} $S$ of a system described by a hamiltonian $H$ is given by
\begin{equation*}
    S(E) = k_{B}\;\ln\;\Omega(E),
\end{equation*}
where $k_{B}$ is Boltzmann's constant. The postulates of Thermodynamics imply that $S$ must be an invertible function and its inverse $U(S)$ is known as the \textbf{internal energy}. From the internal energy we can obtain the \textbf{Helmholtz free energy} $F$ through the Legendre transform, that is
\begin{equation*}
    F = L_{ST}(U).
\end{equation*}
Boltzmann's Principle transformed to the free energy representation states that
\begin{equation}\label{BoltzmannTrans}
    F(T) \approx -k_{B}T\;\ln\;\mathcal{Z}\left(\frac{1}{k_{B}T}\right),
\end{equation}
written more simply as
\begin{equation*}
    F = -k_{B}T\;\ln\;\mathcal{Z}(\beta)
\end{equation*}
where $\beta = \frac{1}{k_{B}T}$. In order to show that the microcanonical measure we just developed reproduces classical statistical physics, we will show that the transformed Boltzmann Principle can be obtained through an asymptotic expansion.

We begin with the rigorous definition of the Legendre transform (see \cite{Legendre}).

\begin{definition}
Let $f:A\subseteq \mathbb{R} \to \mathbb{R}$ be a function that is either concave or convex. We define the \textbf{Legendre transform} $L_{xp}f:A^*\subseteq\mathbb{R}\to \mathbb{R}$ as
\begin{equation*}
    L_{xp}f(p) = \begin{cases}
        \sup\limits_{x\in A}\, f(x)-px &\text{if $f$ is concave},\\
        \\
        \inf\limits_{x\in A}\,f(x)-px &\text{if $f$ is convex},
    \end{cases}
\end{equation*}
where $A^*$ is such that the transform is finite.
\end{definition}
Note also that the subindices of $L_{xp}$ only indicate that the variable $x$ is being changed into the variable $p$. The Legendre transform can also be described as follows (see \cite{LegendreTermo}): If $f$ is twice differentiable then
\begin{equation}\label{EqLegendreEvaluada}
    L_{xp}f(p)= (f(x)-px) \big\vert_{x=(f')^{-1}(p)},
\end{equation}
that is, we must solve $f'(x) = p$ for $x$ and subsitute in $ f(x) - xp$.

The Helmholtz free energy is then defined as
\begin{equation*}
    F(T) = L_{ST}U(T), \quad\text{or simply}\quad F = L_{ST}U,
\end{equation*}
where $U(S)$ is the inverse of $S(E) = k_{B}\;ln\;\Omega(E)$. The last assumption we need to make follows from the postulates of Thermodynamics, and is stated simply as
\begin{equation}\label{EcEstadoT}
    \frac{dS}{dE} = \frac{1}{T},
\end{equation}
to which we will refer as the \textbf{thermal state equation}. Together with equation \eqref{EqLegendreEvaluada} we find that in the Legendre transform of $S$, the parameter $p$ must actually be $\frac{1}{T}$, instead of an arbitrary parameter. Likewise, in the Legendre transform of $U$ the parameter $p$ must be $T$, hence $F$ must have $T$ as its variable.

\subsection{Transformation of the Principle}

We now transform Botlzmann's Principle to show that $F(T)$ can be approximated by $-k_{B}T\;\ln\;\mathcal{Z}$. In the following, we will always assume that $\beta = \frac{1}{k_{B}T}$.

By definition (Boltzmann's Principle) $S = k_{B}\;\ln\;\Omega$. In order to obtain $F$ from $S$ we must first compute the inverse of $S$ then compute its Legendre transform. The following result will allow us to skip this process.

\begin{theorem}
Let $f:A\subseteq\mathbb{R}\to \mathbb{R}$ be a function that is either convex or concave and $a >0$, let $M_a$ denote the function that multiplies by the constant $a$, that is, 
$M_a(x)=ax$.
Then for the function $f\circ M_a$ it holds that
\begin{equation*}
    L_{xp}(f\circ M_a)(p)=L_{x\frac{p}{a}}f\left(\frac{p}{a}\right).
\end{equation*}

Furthermore, if $f$ is invertible and $p>0$ then
\begin{equation*}
    L_{yp}f^{-1}(p)=-pL_{x\frac{1}{p}}f\left(\frac{1}{p}\right).
\end{equation*}
\end{theorem}
\begin{proof}
Suppose $f$ is concave, then $f\circ M_a$ is also concave and by definition
\begin{align*}
    L_{xp}(f\circ M_a)(p)&=\sup_{x\in \frac{1}{a}A } \,f(ax)-px \\&=\sup_{x\in \frac{1}{a}A} \,f(ax)-\frac{p}{a}ax\\&=\sup_{x\in A} \,f(x)-\frac{p}{a}x\\&=L_{x\frac{p}{a}}f\left(\frac{p}{a}\right).
\end{align*}
On the other hand, if we assume that $f$ is invertible and $p>0$, then $f^{-1}$ is convex and by definition
\begin{align*}
    L_{yp}f^{-1}(p)&= \inf_{y\in f(A)} \,f^{-1}(y)-py\\&= \inf_{x\in A} \,f^{-1}(f(x))-pf(x)\\&= \inf_{x\in A} \,x-pf(x)\\&=\inf_{x\in A} \, -p(f(x)-\frac{x}{p})\\&=-p( \sup_{x\in A}\,f(x)-\frac{1}{p}x)\\&=-pL_{xp}f\left(\frac{1}{p}\right).
\end{align*}
The case in which $f$ is convex is similar, interchanging infima and suprema.
\end{proof}

If we apply the previous theorem to the function $\frac{S(E)}{k_{B}}$ with inverse $(U\circ M_{k_B})(s)=U(k_Bs)$  and which according to the thermal state equation \eqref{EcEstadoT} has derivative $\frac{d}{dE}(\frac{S}{k_{B}})=\frac{1}{k_B T}=\beta$, we find that
\begin{align*}
    L_{U \beta}\left(\frac{S}{k_B}\right)(\beta)&=L_{U \beta}\left(\ (U \circ M_{k_B})^{-1}\right)(\beta)\\
    &=-\beta L_{S \frac{1}{\beta}}(U\circ M_{k_{B}})\left(\frac{1}{\beta}\right)\\
    &=-\beta L_{S \frac{1}{k_{B}\beta}}U\left(\frac{1}{k_{B}\beta}\right)\\
    &=-\beta L_{ST}U\left(T\right)\\
    &=-\beta F(T).
\end{align*}
Thus
\begin{equation}\label{trans_feo}
    L_{U \beta}\left(\frac{S}{k_B}\right)(\beta)=-\beta F(T).
\end{equation}

On the other hand, substituting the definition of $S$ given by Boltzmann's Principle we find that
\begin{align*}
   L_{U \beta}\left(\frac{S}{k_B}\right)(\beta)&=L_{U \beta}\left(\ln (\Omega)\right)(\beta)\\&=\sup_{E}\,\ln\,\Omega(E)-E \beta\\&= \sup_{E}\,\ln(\Omega(E)e^{-\beta E})\\&=\ln(\sup_{E}\,\Omega(E)e^{-\beta E}).
\end{align*}
Together with \eqref{trans_feo} we find that
\begin{equation}\label{feo=ln(sup)}
     -\beta F(T) = \ln(\sup_{E}\,\Omega(E)e^{-\beta E}).
\end{equation}
There are now two possible cases: the case in which the function $\Omega(E)e^{-\beta E}$ reaches its maximum value at some energy $E_0>0$, and the case in which it reaches it at energy $E_0=0$. Assuming the case where $E_0>0$, it follows that
\begin{align*} 
\frac{d}{dE}(\Omega(E) e^{-\beta E})\Big|_{E_0} &= g(E_0) e^{-\beta E_0}-\beta \Omega(E_0)e^{-\beta E_0} \\
&= 0, 
\end{align*} 
from which 
\begin{align} 
\frac{g(E_0)}{\Omega(E_0)} = \beta. \label{betaEnE_0} 
\end{align} 
We define the auxiliary function 
\begin{align*} 
\psi(E) \coloneqq \ln(\Omega(E))-\beta E. 
\end{align*} 
With this definition, the canonical partition function can be written as 
\begin{equation*} 
\mathcal{Z}(\beta)= \int_{0}^{\infty} e^{\psi(E)} \,dE. 
\end{equation*} 
Now, we perform a second-order Taylor expansion centered at $E_0$ of the function $\psi$ to approximate the canonical partition function. To do this, we calculate the first and second derivatives of $\psi$ at $E_0$ and use the result \eqref{betaEnE_0}: 
\begin{align*} 
\frac{d\psi}{dE}(E_0) &= \frac{g(E_0)}{\Omega(E_0)}-\beta \\
&= \beta -\beta \\
&= 0. \\ 
\frac{d^2\psi}{dE^2}(E_0) &= \frac{g'(E_0)\Omega(E_0)-g^2(E_0)}{\Omega^2(E_0)} \\
&= \frac{g'(E_0)}{\Omega(E_0)}-\left(\frac{g(E_0)}{\Omega(E_0)}\right)^2 \\
&= \frac{g'(E_0)}{\Omega(E_0)}-\beta^2. 
\end{align*} 
The fact that the first derivative of $\psi$ vanishes at $E_0$ is a consequence of $E_0$ maximizing the product $\Omega(E)e^{-\beta E}$ and therefore also maximizing $\psi$. Consequently, for large values of $\beta$, the function $\psi$ can be approximated as 
\begin{align*} 
\psi(E) \approx \psi(E_0) - \frac{\beta^2}{2}(E-E_0)^2. 
\end{align*} 
By expanding this approximation in the canonical partition function $\mathcal{Z}$ and applying the change of variable given by the function 
\begin{align*} 
k(E) = \frac{\beta}{\sqrt{2}}(E-E_0), 
\end{align*} 
we obtain 
\begin{align*} 
\mathcal{Z}(\beta) &\approx \int_{0}^{\infty} e^{(\psi(E_0) - \frac{\beta^2}{2}(E-E_0)^2)} \,dE \\ 
&= e^{\psi(E_0)}\int_{0}^{\infty} e^{- \frac{\beta^2}{2}(E-E_0)^2} \,dE \\ 
&= \Omega(E_0) e^{-\beta E_0} \, \frac{\sqrt{2}}{\beta} \int_{0}^{\infty} e^{-k^2(E)} k'(E) \,dE \\ 
&= \sup_{E}\,\Omega(E)e^{-\beta E} \, \frac{\sqrt{2}}{\beta} \int_{-\frac{E_0 \beta}{\sqrt{2}}}^{\infty} e^{-k^2} \,dk \\ 
&\approx \sup_{E}\,\Omega(E)e^{-\beta E} \, \frac{\sqrt{2}}{\beta} \int_{-\infty}^{\infty} e^{-k^2} \,dk \\ 
&= \sup_{E}\,\Omega(E)e^{-\beta E} \, \frac{\sqrt{2\pi}}{\beta}. 
\end{align*} 
That is, 
\begin{equation} 
\sup_{E}\,\Omega(E)e^{-\beta E} \approx \frac{\beta}{\sqrt{2\pi}}\mathcal{Z}(\beta). \label{aprox Z gauss} 
\end{equation} 
From equation \eqref{feo=ln(sup)} and the previous approximation \eqref{aprox Z gauss}, we can establish that 
\begin{align*} 
-\beta F(T) &\approx \ln\left(\frac{\beta}{\sqrt{2\pi}}\mathcal{Z}(\beta)\right) \\
&= \ln(\mathcal{Z}(\beta))+\ln\left(\frac{\beta}{\sqrt{2\pi}}\right). 
\end{align*} 
As a result, when the maximum of the function $\Omega(E)e^{-\beta E}$ is reached at $E_0>0$, we have 
\begin{align} 
F(T) \approx -k_B T \,\ln(\mathcal{Z})-\frac{1}{\beta}\, \ln\left(\frac{\beta}{\sqrt{2\pi}}\right).\label{feo_con_E>0} 
\end{align} 

Let us now assume the case where the maximum of the function $\Omega(E)e^{-\beta E}$ is reached at $E_0=0$. We will use an asymptotic result for integrals known as Laplace's approximation method, the proof of which can be found in \cite{Bender}.

\begin{theorem}[Laplace's Approximation] 
Consider a continuous function $f\colon [a,\infty) \to \mathbb{R}$ and a continuously differentiable function $\phi\colon [a,\infty) \to \mathbb{R}$ such that 
\begin{equation*} 
\int_{a}^{\infty} f(E) e^{\beta \phi(E)}\;dE < \infty. 
\end{equation*} 
We define the \textbf{Laplace integral} $I(\beta)$ as 
\begin{equation*} 
I(\beta) = \int_{a}^{\infty} f(E) e^{\beta \phi(E)}\;dE. 
\end{equation*} 
Suppose that $\phi$ reaches its maximum at $E_{0}=a$ and that $f(E_0)\neq0$. If $\phi'(E_0)\neq 0$, we define the \textbf{Laplace approximation} $I_{E_0}(\beta)$ as 
\begin{equation*} 
I_{E_0}(\beta) = - \frac{f(E_0)e^{\beta \phi(E_0)}}{\beta \phi'(E_0)}. 
\end{equation*} 
It holds that 
\begin{equation*} 
\lim_{\beta\to\infty} \frac{I(\beta)}{I_{E_0}(\beta)} = 1. 
\end{equation*} 
\end{theorem}

In our case, the energy $E_0=0$ maximizes the function $\phi(E)=-E$ on the interval $[0,\infty)$. Furthermore, the first derivative $\phi'(E_0)=-1$ is non-zero. Therefore, by applying Laplace's approximation method, we find that for large values of $\beta$: 
\begin{align*} 
\mathcal{Z}(\beta) &= \int_{0}^{\infty} e^{-\beta E}\Omega(E) \,dE \\ 
&\approx - \frac{\Omega(E_0)e^{-\beta E_0}}{\beta \phi'(E_0)} \\ 
&= \frac{1}{\beta} \sup_{E}\,\Omega(E)e^{-\beta E}. 
\end{align*} 
That is, 
\begin{equation} 
\sup_{E}\,\Omega(E)e^{-\beta E} \approx \beta\mathcal{Z}(\beta). \label{aprox Z Laplace} 
\end{equation}
Following equation \eqref{feo=ln(sup)} and the previous approximation \eqref{aprox Z Laplace}, we can establish that 
\begin{align*}
    -\beta F(T) &\approx \ln(\beta\mathcal{Z}(\beta)) \\
    &= \ln(\mathcal{Z}(\beta))+\ln(\beta).
\end{align*}
Therefore, when the maximum of the function $\Omega(E)e^{-\beta E}$ is reached at $E_0=0$, we have
\begin{align}
    F(T) \approx -k_B T \,\ln(\mathcal{Z})-\frac{1}{\beta}\, \ln(\beta).\label{feo_con_E=0}
\end{align}
In either case, \eqref{feo_con_E>0} or \eqref{feo_con_E=0}, for large values of $\beta$ the second term is negligible since
\begin{equation*}
    \lim_{\beta\to\infty}\frac{\ln(\beta)}{\beta} = 0,
\end{equation*}
and the asymptotic equality
\begin{equation*}
    F = -k_{B}T \;\ln(\mathcal{Z})
\end{equation*}
is established.

\subsection{Alternative Definitions}

In physics literature the expression for $\mathcal{Z}$ differs from the one deduced here by a factor of
\begin{equation*}
    \frac{1}{N!h^{3N}},
\end{equation*}
where $N$ is the number of particles of the system and $h$ is planck's constant. This factor is \textquote{added by hand} to avoid Gibbs' paradox (see \cite{Pathria}). This factor can be added to our definition of $\Omega(E)$ without changing the results we obtained, since it is a constant factor that does not alter invariance. This change would carry over to our definition of $\mathcal{Z}$ since the Laplace transform is linear, obtaining the modified expression found in physics literature.

\section{Concluding Remarks}\label{SecConclusiones}

The microcanonical measure solves the problem posed in this paper and allows the use of probability theory for studying hamiltonian systems, reproduces classical Statistical Physics and provides a solid mathematical foundation to the discipline. The problem of short times can be easily dealt with through differential-geometric methods, which are also useful for obtaining results for long times, if regularity hypotheses are imposed. It could be expected that the methods of Geometric Measure Theory could yield similar results without imposing regularity properties on the hamiltonian. It remains to be seen if this is indeed posible, which is left to future reasearch. 

Another road for further investigation is to study whether the measure is conservative (in the sense of Ergodic Theory) or ergodic, or conditions under which either property is achieved. This would allow the application of the classical Ergodic Theorems to obtain further results. This is also of great interest since it is directly related to the Ergodic Hypothesis. These questions are, however, much more subtle and require careful consideration.

Finally, while we did not solve the specific problem posed in Simon's second problem, the results of this paper show that ergodicity is not needed to formulate classical Statistical Mechanics in a mathematically rigorous way. Thus, this work offers an alternative solution to Simon's fundamental question, not by proving ergodicity, but by showing it is not necessary.

\section*{Acknowledgements}

This work was partially supported by DGAPA UNAM through project PAPIME PE100726.\\

\noindent \textbf{Data Availability} 
Data sharing is not applicable to this article as no new data were created or analyzed in this study.
\\

\noindent
\textbf{Declarations}\\

\noindent
\textbf{Conflict of interest} The authors declare that they have no conflict of interest.

\nocite{*}
\bibliographystyle{abbrv}
\bibliography{main}
\addcontentsline{toc}{section}{Bibliography}

\end{document}